%% file: main.tex
\documentclass[11pt]{article}

\usepackage{fullpage}

\usepackage[utf8]{inputenc} 
\usepackage[T1]{fontenc}    
\usepackage{hyperref}       
\usepackage{url}            
\usepackage{booktabs}       
\usepackage{amsfonts}       
\usepackage{nicefrac}       
\usepackage{microtype}      
\usepackage{amssymb}
\usepackage{amsmath}
\usepackage{amsthm}
\usepackage[numbers]{natbib}
\usepackage{wrapfig}
\usepackage{color}
\usepackage{multirow}
\usepackage{rotating}
\usepackage{bbm}
\usepackage{enumitem}
\usepackage{subcaption}
\usepackage[noend]{algpseudocode}
\usepackage{algorithm}
\usepackage{eqparbox,array}
\usepackage{setspace}

\algrenewcommand{\algorithmiccomment}[1]{\hfill// \eqparbox{COMMENT\thealgorithm}{\footnotesize{\ttfamily#1}}}
\algnewcommand{\LongComment}[1]{\hfill// \begin{minipage}[t]{\eqboxwidth{COMMENT\thealgorithm}}#1\strut\end{minipage}}

\usepackage{bm}
\usepackage{tcolorbox}
\usepackage{hhline}
\usepackage{array}

\numberwithin{equation}{section}

\input{notation}

\makeatletter
\newtheorem*{rep@theorem}{\rep@title}
\newcommand{\newreptheorem}[2]{%
	\newenvironment{rep#1}[1]{%
		\def\rep@title{#2 \ref{##1}}%
		\begin{rep@theorem}}%
		{\end{rep@theorem}}}
\makeatother

\newtheorem{theorem}{Theorem}
\newtheorem{lemma}{Lemma}
\newreptheorem{claim}{Claim}
\newreptheorem{fact}{Fact}
\newtheorem{definition}{Definition}
\newtheorem*{thm*}{Theorem}

\usepackage{xcolor}

\usepackage{thm-restate}

\title{Learning-Augmented Dynamic Submodular Maximization}

\author{%
  Arpit Agarwal\thanks{The author is currently with FAIR, Meta. Work done while the author was at Columbia University.} \\
  Columbia University\\
  New York, NY 10027, USA \\
  \texttt{agarpit@outlook.com}
  \and 
    Eric Balkanski \\
  Columbia University\\
  New York, NY 10027, USA \\
  \texttt{eb3224@columbia.edu }
}
\date{}

\begin{document}

\maketitle

\begin{abstract}
 In dynamic submodular maximization, the goal is to maintain a high-value solution over a sequence of element insertions and deletions  with  a fast update time. Motivated by large-scale applications and the fact that dynamic data often exhibits patterns, we ask the following question: can predictions be used to accelerate the update time of dynamic submodular maximization algorithms? 

We consider the  model for dynamic algorithms with predictions where  predictions regarding the insertion and deletion times of elements  can be used for preprocessing. Our main result is an algorithm with an $O(\poly(\log \eta, \log w, \log  k))$ amortized update time over the sequence of updates 
that achieves a $1/2 - \epsilon$ approximation in expectation for dynamic monotone submodular maximization under a cardinality constraint $k$, where the prediction error $\eta$ is the number of elements that are not inserted and deleted within $w$ time steps of their predicted insertion and deletion times. This amortized update time is independent of the length of the stream and instead depends on the prediction error.
\end{abstract}
\thispagestyle{empty} 
\newpage

\section{Introduction}
\setcounter{page}{1} 
\input{1.intro}

\section{Preliminaries}

\input{2.prelim}

\section{The Warm-Up Algorithm}
\label{sec:warm-up}

\input{3.warm-up}

\section{The UpdateSol Subroutine}
\label{sec:updatesol}
\input{4.update_time}

\section{The Precomputations Subroutine}
\label{sec:precomp}

\input{5.precomputations}


\section{The Full Algorithm}
\label{sec:mainresult}
\input{6.main_result}

\section*{Acknowledgements}

Eric Balkanski was supported by NSF grants CCF-2210502 and IIS-2147361.

\newpage

\bibliographystyle{plainnat}
\bibliography{refs}

\newpage

\appendix

\section{Missing Proofs from Section~\ref{sec:warm-up}}
\label{app:warmup}

\input{appendix_warmup}

\input{appendix_updatetime}

\input{appendix_precomp}

\end{document}

%% file: notation.tex
\newcommand{\qe}{\nobreak \ifvmode \relax \else
      \ifdim\lastskip<1.5em \hskip-\lastskip
      \hskip1.5em plus0em minus0.5em \fi \nobreak
      \vrule height0.75em width0.5em depth0.25em\fi}

\newcount\Comments  
\newcommand{\kibitz}[2]{\ifnum\Comments=1{\color{#1}{#2}}\fi}

\usepackage{mathtools}

\newcommand{\OPT}{\texttt{OPT}}

\renewcommand{\>}{\rightarrow}

\newcommand{\beq}{\begin{equation}}
\newcommand{\eeq}{\end{equation}}

\newcommand{\bea}{\begin{equation}}
\newcommand{\eea}{\end{equation}}

\newcommand{\bes}{\begin{split}}
\newcommand{\ees}{\end{split}}

\newcommand{\benum}{\begin{enumerate}}
\newcommand{\eenum}{\end{enumerate}}

\newcommand{\heta}{\eta_{\text{old}}}

\newcommand{\sol}{\ensuremath\texttt{SOL}}

\newcommand{\NULL}{\ensuremath\texttt{null}}
\newcommand{\robustone}{\textsc{Robust1}}

\newcommand{\updatesol}{\textsc{UpdateSol}}

\newcommand{\robusttwo}{\textsc{Robust2}}
\newcommand{\robust}{\textsc{Robust}}
\newcommand{\dynamic}{\textsc{Dynamic}}

\newcommand{\dynamicins}{\textsc{DynamicIns}}

\newcommand{\dynamicdel}{\textsc{DynamicDel}}
\newcommand{\dynamicinit}{\textsc{DynamicInit}}
\newcommand{\elem}{\textsc{elem}}
\newcommand{\dynamicsol}{\textsc{DynamicSol}}

\newcommand\R{\mathbb{R}}

\newcommand\1{\mathbbm{1}}

\newcommand{\hV}{\hat{V}}

\newcommand{\Z}{\mathbb{Z}}

\newcommand{\E}{\mathbf{E}}

\DeclareMathOperator{\argmax}{argmax}

\newtheorem{lem}{Lemma}

\newcommand {\commentout}[1] {}

\newcommand{\poly}{\ensuremath\textnormal{poly}}
\newcommand{\polylog}{\ensuremath\textnormal{polylog}}

\newenvironment{tbox}{\begin{tcolorbox}[
		enlarge top by=5pt,
		enlarge bottom by=5pt,
		 breakable,
		 boxsep=0pt,
                  left=4pt,
                  right=4pt,
                  top=10pt,
                  arc=0pt,
                  boxrule=1pt,toprule=1pt,
                  colback=white
                  ]
	}
{\end{tcolorbox}}


%% file: 1.intro.tex
Submodular functions are a well-studied family of functions that satisfy a natural diminishing returns property. Since many fundamental objectives  are submodular, including coverage, diversity, and entropy, submodular optimization algorithms play an important role in machine learning~\cite{krause2014submodular, bilmes2022submodularity}, network analysis~\cite{kempe2003maximizing}, and mechanism design~\cite{singer2010budget}.   For the canonical problem of maximizing a monotone submodular function under a cardinality constraint, the celebrated greedy algorithm achieves a $1-1/e$ approximation guarantee~\cite{NWF78}, which is the best approximation guarantee achievable by any polynomial-time algorithm \cite{nemhauser1978best}.

Motivated by the highly dynamic nature of applications such as influence maximization in social networks and recommender systems on streaming platforms, a  recent line of work has studied the problem of  dynamic submodular maximization \cite{lattanzi2020fully,Monemizadeh20,peng2021dynamic,chen2022complexity, dutting2023fully,Banihashem+23, banihashem2024dynamic}.  In the dynamic setting, the input consists of a stream of elements that are inserted or deleted from the set of active elements, and the goal is to maintain, throughout the stream, a subset of the active elements that maximizes a submodular function.   

The standard worst-case approach to analyzing the update time of a dynamic algorithm is to measure its update time over the worst sequence of updates possible. However, in many application domains, dynamic data is not arbitrary and often exhibits patterns that can be learned from historical data.  Very recent work has studied dynamic problems in settings where the algorithm is given as input predictions regarding the stream of updates \cite{henzinger2023complexity,liu2023predicted,brand2023dynamic}.  This recent work is part of a broader research area called learning-augmented algorithms (or algorithms with predictions). In learning-augmented algorithms, the goal is to design algorithms that  
achieve an improved performance guarantee when the error of the prediction is small and a bounded guarantee even when the prediction error is arbitrarily large. A lot of the effort in this area has been focused
on using predictions to improve the competitive ratio of
online algorithms (see, e.g., \cite{lykouris2018competitive,purohit2018improving, wang2020online,purohit2018improving,mitzenmacher2020scheduling,dutting2021secretaries, antoniadis2020secretary,bamas2020primal,lattanzi2020online,im2021online,banerjee2020online}), and more generally to improve the
solution quality of algorithms. 

For dynamic submodular maximization  with predictions, Liu and Srinivas~\cite{liu2023predicted} considered a predicted-deletion model in which they achieved, under some mild assumption on the prediction error, a $0.3178$ approximation and an $\tilde{O}(\poly(k, \log n))$\footnote{In this paper, we use the notation $\tilde{O}(g(n))$ as shorthand for $O(g(n) \log^kg(n))$.} update time for dynamic monotone submodular maximization under a matroid constraint of rank $k$ and over a stream of length $n$. This approximation is an improvement over the best-known $1/4$ approximation for dynamic monotone submodular maximization (without predictions) under a matroid constraint with an update time that is sublinear in $n$ \cite{banihashem2024dynamic, dutting2023fully}. Since the update time of dynamic algorithms is often the main bottleneck in large-scale problems, another promising direction is to leverage predictions to improve the update time of dynamic algorithms.

\begin{center}
\emph{Can predictions be used to accelerate the update time  of dynamic \\ submodular maximization algorithms?}
\end{center}

We note that the three very recent papers on dynamic algorithms with predictions have achieved improved update times  for several dynamic graph problems \cite{henzinger2023complexity,liu2023predicted,brand2023dynamic}. However, to the best of our knowledge, there is no previous result that achieves an improved update time  for dynamic submodular maximization by using predictions.

\subsection{Our contributions}

In  dynamic submodular maximization, the input is a submodular function $f~:~2^V \rightarrow \R_{\geq 0}$ and a sequence of $n$ element insertions and deletions. The active elements $V_t \subseteq V$ at time $t$ are the elements that have been inserted and have not yet been deleted during the first $t$ updates. The goal is to maintain, at every time step $t$, a solution $S_t \subseteq V_t$ that is approximately optimal with respect to $V_t$ while minimizing the number of queries to $f$ at each time step, which is referred to as the update time.  As in~\cite{henzinger2023complexity, brand2023dynamic},
we consider a prediction model where, at time $t = 0$, the algorithm is given a prediction regarding the sequence of updates, which can be used for preprocessing. More precisely, at time $t = 0$, the algorithm is given predictions $(\hat{t}^+_a, \hat{t}^-_a)$ about the insertion and deletion time of each element $a$.  A dynamic algorithm with predictions then consists of a precomputation phase at  $t = 0$, where the algorithm uses the predictions to perform queries before the start of the stream, and a streaming phase at time steps $t > 0$, where the algorithm performs queries, and uses the precomputations, to maintain a good solution with respect to the true stream. 

In this model, there is a trivial algorithm that achieves a constant update time when the predictions are exactly correct  and an  $O(u)$ update time when the predictions are arbitrarily wrong, where $u$ is the update time of an arbitrary  algorithm $\mathcal{A}$ for the problem without predictions. This algorithm precomputes, for each future time step $t$, a solution for the elements that are predicted to be active at time $t$ and then, during the streaming phase, returns the precomputed solutions while the prediction is correct and switches to running algorithm $\mathcal{A}$ at the first error in the prediction. Thus, the interesting question is whether it is possible to obtain an improved update time not only when the predictions are exactly correct, but  more generally when the error in the predictions is small. An important component of our model is  the measure for the prediction error. Given a time window tolerance $w$, an element $a$ is considered to be correctly predicted if the predicted insertion and deletion times of $a$ are both within $w$ times steps of its true insertion and deletion times. The prediction error $\eta$ is then the number of elements that are not correctly predicted. Thus, $\eta = 0$ if the predictions are exactly correct and $\eta = \Theta(n)$ if the predictions are completely wrong.

For dynamic monotone submodular maximization (without predictions) under a cardinality constraint $k$, Lattanzi et al.~\cite{lattanzi2020fully} and Monemizadeh~\cite{Monemizadeh20} concurrently obtained dynamic algorithms with $O(\poly(\log n, \log k))$ and $O(\polylog(n) \cdot  k^2)$  amortized update time, respectively, that achieve a $1/2 - \epsilon$ approximation in expectation. More recently, Banihashem et al.~\cite{banihashem2024dynamic} achieved a $1/2 - \epsilon$ approximation  with a $O(k \cdot \polylog(k))$ amortized update time.   Our main result is the following.

\begin{thm*}
For monotone submodular maximization under a cardinality constraint $k$,  
there is a dynamic algorithm with predictions that, for any tolerance $w$ and constant $\epsilon > 0$, achieves  an amortized query complexity during the streaming phase of $O(\poly(\log \eta, \log w, \log k))$,  
 an approximation of $1/2 - \epsilon$ in expectation, and a query complexity of $\tilde{O}(n)$ during the precomputation phase. 
\end{thm*}

We note that, when the prediction error $\eta$ is arbitrarily large, our algorithm   matches  the $O(\poly(\log n, \log k))$ amortized query complexity from \cite{lattanzi2020fully}. It also achieves an approximation ratio that matches the optimal  approximation for dynamic algorithms (without predictions) with update time that is sublinear in the length of the stream. An intriguing open question is, in our view, whether a similar improvement in update time can be obtained in the more challenging predicted-deletion model of \cite{liu2023predicted} where there is no preprocessing and instead a predicted deletion time for each element $a$ is given at the time when $a$ is inserted.

\subsection{Related work}

\paragraph*{Dynamic submodular maximization} 

For monotone submodular maximization under a cardinality constraint,  dynamic algorithms with $O(\poly(\log n, \log k))$ and $O(\polylog(n) \cdot  k^2))$  amortized update time that achieve a $1/2$ approximation  were concurrently obtained in  \cite{lattanzi2020fully, Monemizadeh20}. Recently, Banihashem et al.~\cite{banihashem2024dynamic} gave a $1/2$ approximation algorithm with a $O(k \cdot \polylog(k))$ amortized update time.  Chen and Peng~\cite{chen2022complexity}  showed that any dynamic  algorithm  with an approximation  better than $1/2$ must have $\poly(n)$ amortized query complexity. For matroid constraints, Chen and Peng~\cite{chen2022complexity} obtained an insertion-only  algorithm. As mentioned in \cite{peng2023fully}, the streaming algorithm in \cite{feldman2022streaming} can  be adapted to also give an insertion-only  algorithm.  Two   $1/4$-approximation  dynamic algorithms   with $\tilde{O}(k)$ and $\tilde{O}(\polylog(n) \cdot k^2)$ amortized update time were concurrently obtained in \cite{banihashem2024dynamic} and \cite{dutting2023fully}.

\paragraph*{Learning-augmented algorithms}

A main focus in algorithms with predictions has been to improve the competitive ratio of online  problems (see, e.g., \cite{lykouris2018competitive, purohit2018improving}). Leveraging predictions to improve
the solution quality  has also been studied in mechanism design \cite{ agrawal2022learning,xu2022mechanism} and differential privacy~\cite{amin2022private}. Improved running times  have been obtained with predictions for  matching \cite{dinitz2021faster}, hashing \cite{andoni2022learning}, and clustering \cite{ergun2021learning}. 
 Improved update times  for several dynamic graph problems were very recently obtained by leveraging predictions \cite{henzinger2023complexity,liu2023predicted,brand2023dynamic}. In particular, 
Liu and Srinivas~\cite{liu2023predicted}  obtained, under some mild assumption on the prediction error, a $0.3178$ approximation and a $\tilde{O}(\poly(k, \log n))$ update time for dynamic monotone submodular maximization under a matroid constraint of rank $k$ in the more challenging predicted-deletion model. Thus, by using predictions, this result improves the $1/4$ approximation achieved, without predictions, in \cite{banihashem2024dynamic} and \cite{dutting2023fully}  (but does not improve the $1/2$ approximation for cardinality constraints). The results in \cite{liu2023predicted} use a framework that takes as input an insertion-only dynamic algorithm. In contrast, we develop a framework that uses a fully dynamic algorithm and a deletion-robust algorithm.

%% file: 2.prelim.tex
A function $f : 2^V \rightarrow \R$ defined over a ground set $V$ is submodular if  for all $S \subseteq T \subseteq V$ and $a \in V \setminus T$, we have that $f_S(a) \geq f_T(a)$, where $f_S(a) = f(S \cup \{a\}) - f(S)$ is the marginal contribution of $a$ to $S$. It is monotone if $f(S) \leq f(T)$ for all $S \subseteq T \subseteq V$. We consider the canonical problem $\max_{S \subseteq V: |S| \leq k} f(S)$
of maximizing a monotone submodular function $f$ under a cardinality constraint $k$.

\paragraph*{Dynamic submodular maximization}

 In dynamic submodular maximization,  there is a stream $\{(a_t, o_t)\}_{t = 1}^n$ of $n$ element insertions and deletions where $o_t \in \{\text{insert}, \text{deletion}\}$ and $a_t$ is an element in $V$. The active elements $V_t$  are the elements that have been inserted and have not been deleted by time $t$. We assume that $(a_t, \text{insertion})$ and $(a_t, \text{deletion})$ can occur in the stream only if $a_t \not \in V_t$ and $a_t \in V_t$, respectively,  and that each element is inserted at most once.\footnote{If an element $a$ is re-inserted, a copy $a'$ of $a$ can be created.} The goal  is to maintain  a solution $S_t \subseteq V_t$ that approximates the optimal solution over $V_t$. Since our algorithmic framework takes as input a dynamic algorithm, we formally define dynamic algorithms in terms of black-box subroutines that are used in our algorithms.

\begin{definition}
\label{def:dynamic}
Given a function $f : 2^V \rightarrow \R$ and a cardinality constraint $k$, 
a dynamic  algorithm $\dynamic(f, k)$ consists of the following four subroutines such that, for any function $f$,  constraint $k$, and stream $\{(a_t, o_t)\}_{t = 1}^n$,
\begin{itemize}
\item $\textsc{DynamicInit}(f, k)$ initializes a data structure $A$ at $t = 0$, 
\item if $o_t = \text{insertion}$, $\dynamicins(A, a_t)$ inserts element $a_t$ in  $A$ at time $t$,
\item if $o_t = \text{deletion}$, $\dynamicdel(A, a_t)$ deletes element $a_t$ from $A$ at time $t$, 
\item at  time  $t$,  $\dynamicsol(A)$ returns  $S_t \subseteq V(A)$ s.t. $|S| \leq k$, where $V(A) = V_t$ is the set of elements that have been inserted and not been deleted from $A$.
\end{itemize}
\end{definition}

A dynamic algorithm achieves an $\alpha$-approximation if, for all time steps $t$, $\E[f(S_t)] \geq \alpha \cdot \max_{S \subseteq V_t: |S| \leq k} f(S)$ and  has a $u(n, k)$ amortized query complexity per update if its worst-case total number of queries is $n \cdot  u(n,k).$

\paragraph*{Dynamic submodular maximization with predictions}

In dynamic submodular maximization with predictions, the algorithm is given at time $t = 0$  predictions $\{(\hat{t}^+_a, \hat{t}^-_a)\}_{a \in V}$ about the insertion and deletion time of each element $a \in V$. The prediction error $\eta$ is the number of elements that are incorrectly predicted, where an element $a$ is correctly predicted if it is inserted and  deleted within a time window, of size parameterized by a time window tolerance $w$, that is centered at the time at which $a$ is predicted to be inserted and  deleted. 

\begin{definition}
\label{def:error}
     Given a time window tolerance $w \in \Z_+$, predictions  $\{(\hat{t}^+_a, \hat{t}^-_a)\}_{a \in V}$, and true insertion and deletions times $\{(t^+_a, t^-_a)\}_{a \in V}$, the prediction error  is $$\eta = |\{a \in V: |\hat{t}^+_a - t^+_a| > w \text{ or } |\hat{t}^-_a - t^-_a| > w\}|.$$
\end{definition}

We note that $\eta = w = 0$ corresponds to the predictions all being exactly correct and that $\eta = O(n)$ and $w = O(n)$ corresponds to the predictions being arbitrarily wrong.

\paragraph*{Deletion-robust submodular maximization}

 Our framework also takes as input a \emph{deletion-robust algorithm}, which we  formally define in terms of black-box subroutines. A deletion-robust algorithm finds a solution $S \subseteq V$ that is robust to the deletion  of at most $d$ elements.

\begin{definition}
\label{def:robust} Given a function $f : 2^V \rightarrow \R$, a cardinality constraint $k$, and a maximum number of deletions parameter $d$,
a deletion-robust algorithm $\robust(f, V, k, d)$ consists of a first subroutine $\robustone(f, V, k, d)$ that returns a robust set $R \subseteq V$ and  a second subroutine  $\robusttwo(f, R, D, k)$ that returns a set $S \subseteq R \setminus D$ such that $|S| \leq k$.
\end{definition}

A deletion-robust algorithm achieves an $\alpha$ approximation   if, for any $f$, $V$, $k$, and $d$, the subroutine $\robustone(f, V, k, d)$ returns $R $ such that, for any $D \subseteq V$ such that $|D| \leq d$, $\robusttwo(f, R, D, k)$ returns  $S$ such that $\E[f(S)] \geq \alpha \cdot \max_{T \subseteq V \setminus D: |T| \leq k} f(T)$. Kazemi et al.~\cite{kazemi2017deletion} show that there is a  deletion-robust  
 algorithm for monotone submodular maximization under a cardinality constraint such that $\robustone$ returns a set $R$ of size $|R| = O(\epsilon^{-2}d \log k + k)$. It achieves a $1/2 - \epsilon$ approximation in expectation, $\robustone$ has  $O(|V| \cdot (k+\epsilon^{-1}\log k))$ query complexity,  and $\robusttwo$ has $O\left((\epsilon^{-2}d\log k + k) \cdot \epsilon^{-1}\log k\right)$ query complexity.



\subsection{The algorithmic framework}

In this section, we present an algorithmic framework that decomposes dynamic algorithms with predictions into two subroutines, \textsc{Precomputations} and \textsc{UpdateSol}. The remainder of the paper then consists of designing and analyzing these subroutines.

We first introduce some terminology. An element $a$ is said to be correctly predicted if $|\hat{t}^+_a - t^+_a| \leq w \text{ and } |\hat{t}^-_a - t^-_a| \leq w$.
The \emph{predicted elements} $\hat{V}_t$ consist of all elements that could potentially be active at time $t$ if correctly predicted, i.e., the elements $a$ such that $\hat{t}^+_a \leq t+ w \text{ and } \hat{t}^-_a \geq t - w$. 
During the precomputation phase, the first subroutine, \textsc{Precomputations}, takes as input the predicted  elements $\hat{V}_t$  and outputs, for each time step $t$, a data structure $P_t$ that will then be used at time $t$ of the streaming phase to compute a solution efficiently. During the streaming phase, the active elements $V_t$
are partitioned into the \emph{predicted active elements} $V_t^1 = V_t \cap \hat{V}_t$ and the \emph{unpredicted active elements} $V_t^2 = V_t \setminus \hat{V}_t$. The second subroutine, \textsc{UpdateSol}, is given  $V_t^1$ and  $V_t^2$ as input  and computes a solution $S \subseteq  V_t^1 \cup  V_t^2$ at each time step.  \textsc{UpdateSol} is also given as input precomputations $P_t$ and the current prediction error $\eta_t$. It also stores useful information for future time steps in a data structure $A$. 

\begin{algorithm}[H]
\caption{The Algorithmic Framework}
\setstretch{1.08}
\hspace*{\algorithmicindent} \textbf{Input:} function $f : 2^V \> \R$, constraint $k$, predictions $\{(\hat{t}^+_a, \hat{t}^-_a)\}_{a \in V}$, tolerance $w$
\begin{algorithmic}[1]
\State{$\hat{V}_{t} \leftarrow \{a \in V: \hat{t}^+_a \leq t+ w \text{ and } \hat{t}^-_a \geq t - w\}$ for $t \in [n]$}
\State{$\{ P_t\}_{t=1}^n \leftarrow \textsc{Precomputations}(f, \{\hV_i\}_{t=1}^n, k)$}
\State{$V_0, A \leftarrow \emptyset$}
\For{$t=1$ to $n$} 
\State{Update active elements $V_t$ according to operation at time $t$}
\State{$V^1_t \leftarrow  V_t \cap \hV_t$}
\State{$V^2_t \leftarrow V_t \setminus \hV_t$}
\State{$\eta_t \leftarrow $ current prediction error}
\State{$A, S \leftarrow$ \textsc{UpdateSol}$(f, k, A, t, P_t, V^1_t, V^2_t, \hV_{t}, \eta_t)$}
   \State{\textbf{return} $S$ }
\EndFor
\end{algorithmic}
\label{alg:framework}
\end{algorithm}

%% file: 3.warm-up.tex
In this section, we present   subroutines that achieve an $\tilde{O}(\eta +w + k)$ amortized update time and a $1/4- \epsilon$ approximation in expectation. These warm-up subroutines assume that the error $\eta$ is known. They take as input  a dynamic algorithm (without predictions) \textsc{Dynamic} and a deletion-robust algorithm \textsc{Robust} algorithm. Missing proofs are deferred  to Appendix~\ref{app:warmup}.

\paragraph*{The precomputations subroutine}

A main observation is that the problem of finding a solution $S^1 \subseteq V^1_t$ among the  predicted active elements  corresponds to a deletion-robust problem over $\hat{V}_t$ where the  deleted elements $D$ are the predicted elements $\hV_t \setminus V_t$ that  are not active at time $t$. \textsc{WarmUp-Precomputations} thus calls, for each time $t$, the first stage \textsc{Robust1} of \textsc{Robust}.

\begin{algorithm}[H]
\caption{\textsc{WarmUp-Precomputations}}
\setstretch{1.05}
\hspace*{\algorithmicindent} \textbf{Input:} function $f$, predicted elements $\{\hV_t\}$, constraint $k$, prediction error $\eta$
\begin{algorithmic}[1]
\For{$t \in \{1, \ldots, n\}$}
\State{$P_t \leftarrow \robustone(f, \hat{V}_{t}, k, d = \eta + 2w)$}
\EndFor
\State{\textbf{return} $\{ P_t\}_{t=1}^n $}
\end{algorithmic}
\label{alg:WarmUpPre}
\end{algorithm}

When calling \textsc{Robust1}, the algorithm sets
the maximum number of deletions parameter $d$  to $\eta + 2w$
because the  number of predicted elements $\hV_t \setminus V_t$ that  are not active at time $t$ is always at most $\eta + 2w$, as shown by the next lemma.

\begin{restatable}{rLem}{lemstreamrobust}
\label{lem:stream_robust}
  Assume that the predicted stream has, with a time error tolerance $w$, a prediction error at most $\eta$. Then,
at any time step $t$, $|\hat{V}_t \setminus V_t| \leq \eta + 2w$.  
\end{restatable}

\paragraph*{The updatesol subroutine}

\textsc{WarmUp-UpdateSol} finds a solution $S^1 \subseteq V_t^1$  by  calling  \robusttwo \  over the precomputed  $P_{t}$  and deleted elements $D = \hV_t \setminus V_t$. To find a  solution $S^2 \subseteq V^2_t$ among the unpredicted active elements, we use \dynamic \ over the stream of element insertions and deletions that result in unpredicted active elements $V^2_1, \ldots, V^2_n$, which is the stream that inserts elements $V^2_t \setminus V^2_{t-1}$ and deletes elements $V^2_{t-1} \setminus V^2_{t}$ at time $t$. The  solution $S^2$ is then the solution produced by  \dynamicsol \ over this stream. The solution $S$ returned by \textsc{UpdateSol} is   the best solution between $S^1$ and $S^2$.

\begin{algorithm}[H]
\caption{\textsc{WarmUp-UpdateSol}}
\setstretch{1.05}
\hspace*{\algorithmicindent} \textbf{Input:} function $f$, constraint $k$, data structure $A$, time $t$, precomputations $P_t$, predicted active elements $V^1_t$, unpredicted active elements $V^2_t$, predicted elements $\hV_t$
\begin{algorithmic}[1]
\State{$S^1 \leftarrow \robusttwo(f, P_t, \hV_t \setminus V^1_t, k)$} 
\If{ $t = 1$}
\State{$A \leftarrow \textsc{DynamicInit}(f, k)$}
\EndIf
\For{ $a \in  V^2_t \setminus V(A)$}
\State{$\dynamicins(A, a)$}
\EndFor
\For{ $a \in V(A) \setminus V^2_t$}
\State{$\dynamicdel(A, a)$}
\EndFor
  \State{$S^2 \leftarrow \dynamicsol(A)$}
\State{\textbf{return} $A, \argmax\{f(S^1), f(S^2)\}$}
\end{algorithmic}
\label{alg:WarmUpsol}
\end{algorithm}

\paragraph*{The analysis of the warm-up algorithm}

We first analyze the approximation. We let $\alpha_1$ and $\alpha_2$ denote the approximations achieved by \robust \ and \dynamic. The first lemma shows that solution $S^1$ is an $\alpha_1$ approximation to the optimal solution over the predicted active elements $V_t^1$ and that solution $S^2$ is an $\alpha_2$ approximation to the optimal solution over the unpredicted active elements $V_t^2.$ The proof uses Lemma~\ref{lem:stream_robust} and the definitions of deletion-robust and dynamic algorithms.

\begin{restatable}{rLem}{robustwu}
\label{lem:robust_wu}
 At every time step $t$,  $f(S^1)  \geq \alpha_1  \cdot \OPT(V_t^1)$ and $f(S^2) \geq \alpha_2 \cdot \OPT(V_t^2)$.
\end{restatable}

By a standard application of monotonicity and submodularity, the previous lemma implies that $\max\{ f(S^1), f(S^2)\}$ achieves a $\frac{1}{2} \cdot\min\{\alpha_1, \alpha_2\}$ approximation.

\begin{restatable}{rLem}{combinedsol}
\label{lem:combined_sol}
At every time step $t$,  $\max\{ f(S^1), f(S^2)\} \geq \frac{1}{2} \cdot\min\{\alpha_2, \alpha_1\}   \cdot \OPT(V_t)$.
\end{restatable}

The main lemma for the amortized query complexity bounds the number of calls to \dynamicins.

\begin{restatable}{rLem}{sizeA}
\label{lem:size_A}
\textsc{WarmUp-UpdateSol} makes at most $2 \eta$ calls to $\dynamicins$ on $A$ over the stream.
\end{restatable}

The main result for the warm-up algorithm is the following.
\begin{theorem}
\label{thm:warmup-main}
For monotone submodular maximization under a cardinality constraint $k$,  
Algorithm~\ref{alg:framework} with the \textsc{WarmUp-Precomputations} and \textsc{WarmUp-UpdateSol} subroutines achieves, for any tolerance $w$ and constant $\epsilon > 0$,   an amortized query complexity during the streaming phase of $\tilde{O}(\eta + w + k)$,  an approximation of $1/4 - \epsilon$ in expectation, and a query complexity of $\tilde{O}(n^2 k)$ during the precomputation phase. 
\end{theorem}
\begin{proof}
We choose \dynamic \ to be the algorithm from \cite{lattanzi2020fully}  and \robust \ to be the algorithm from \cite{kazemi2017deletion}.  By Lemma~\ref{lem:combined_sol}, and since we have $\alpha_1 = \alpha_2 = 1/2 - \epsilon$, the approximation is $\frac{1}{2} \cdot\min\{\alpha_2, \alpha_1\} = \frac{1}{2} \cdot\min\{\frac{1}{2} - \epsilon, \frac{1}{2} - \epsilon\} \geq \frac{1}{4} - \epsilon$.

  For the query complexity during the streaming phase, by Lemma~\ref{lem:size_A}, the total number of calls to \dynamicins \ on $A$ is $O(\eta)$. The total number of calls to \dynamicdel \ on $A$ is also  $O(\eta)$ since an element can only be deleted if it has been inserted. Thus, the total number of insertions and deletions handled by the dynamic streaming algorithm \dynamic \ is $O(\eta)$. Since the amortized query complexity of \dynamic \ is $O(\poly(\log n, \log k))$, we get that the amortized number of queries performed when calling \dynamicins, \dynamicdel, and \dynamicsol \ on $A$ is $O(\eta \cdot \poly(\log \eta, \log k) / n)$. For the number of queries due to 
$\robusttwo$,
at every time step $t$, the algorithm calls \robusttwo$(f, P_t, (\hat{V}_t \setminus V^1), k)$ with maximum number of deletions $d = \eta + 2w$, which causes at most $\tilde{O}\left(\eta + w + k\right)$ queries at each time step $t$ since the query complexity of $\robusttwo$ is $\tilde{O}\left(d + k\right)$. The total amortized query complexity during the streaming phase is thus $\tilde{O}(\eta + w + k)$. For the query complexity during the precomputation phase,   the query complexity of $\robustone$
    is $\tilde{O}(|V| k)$ and there are $n$ calls to $\robustone$, so the query complexity of that phase is $\tilde{O}(n^2 k)$.
\end{proof}

 In the next sections, we  improve the dependency on $\eta, w, $ and $k$ for the amortized query complexity from linear to logarithmic, the approximation factor from $1/4$ to $1/2$, and the precomputations query complexity from $O(n^2 k)$ to $\tilde{O}(n)$. We also remove the assumption that the prediction error $\eta$ is known.

%% file: 4.update_time.tex
The main technical difficulty is to improve the linear dependencies in $\eta, w,$ and $k$ for the amortized query complexity of \textsc{UpdateSol}. In this section, we improve them from linear to logarithmic. For finding a solution over the predicted active elements $V^1_t$, the main idea is to not only use precomputations $P_t,$ but also to exploit computations from previous time steps $t' < t$ over the previous predicted active elements $V_{t'}^1$. As in the warm-up subroutine, the new \textsc{UpdateSolMain}  subroutine also uses a precomputed deletion-robust solution $P_t$, but it  requires  $P_t$ to satisfy a property termed the \emph{strongly robust property} (Definition~\ref{def:robust_properties} below), which is stronger than the deletion-robust property of Definition~\ref{def:robust}.  A strongly robust solution comprises two components $Q$ and $R$, where  $R$ is a small set of elements that  have a high marginal contribution to $Q$. The set $Q$ is such that, for any deleted set $D$,  $f(Q \setminus D)$ is guaranteed to, in expectation over the randomization of $Q$, retain a large amount of $f(Q)$.

\begin{definition}
\label{def:robust_properties}
A pair of sets $(Q, R)$ is $(d, \epsilon, \gamma)$-strongly robust, where $d,k, \gamma \geq 0$ and $\epsilon \in [0, 1]$, if
\begin{itemize}
\item \textbf{Size.} $|Q| \leq k$ and $|R| = 
O(\epsilon^{-2}(d + k)\log k)$ with probability $1$,
\item \textbf{Value.}  $ f(Q) \geq |Q|  \gamma / (2k)$ with probability $1$.
In addition, if $|Q| < k$, then 
for any set $S \subseteq V \setminus R$ we have $f_Q(S ) < |S| \gamma / (2k) + \epsilon \gamma$.
\item \textbf{Robustness.}  For any  $D \subseteq V$ s.t. $|D| \leq d$, $\E_Q[f(Q \setminus D)] \geq (1-\epsilon)  f(Q).$ 
\end{itemize}
\end{definition}

The set $P$ returned by the first stage $\textsc{Robust1}(f, V, k, d)$ of the deletion-robust algorithm of Kazemi et al.~\cite{kazemi2017deletion} can be decomposed into two sets $Q$ and $R$ that are, for any $d, \epsilon > 0$ and with $\gamma = \OPT(V)$, where $\OPT(V) := \max_{S \subseteq V : |S| \leq k} f(S)$,  $(d, \epsilon, \gamma)$-strongly robust.\footnote{The size of $R$ in \cite{kazemi2017deletion} is $O(d \log k/ \epsilon)$ which is better than what is required to be strongly robust.} Thus, with the $\textsc{Robust}$ algorithm of \cite{kazemi2017deletion}, the set $P_t$ returned by \textsc{WarmUp-Precomputations} can be decomposed into $P_t$ and $Q_t$ that are, for any $\epsilon > 0$, $(2(\eta + 2w), \epsilon, \OPT(\hat{V}_t))$-strongly robust. We omit the proof of the $(d, \epsilon, \gamma)$-strongly robust property of $\textsc{Robust1}$ from \cite{kazemi2017deletion} and, in the next section, we instead prove strong-robustness for our \textsc{PrecomputationsMain} subroutine which has better overall query complexity than \cite{kazemi2017deletion}.


The \textsc{UpdateSolMain}  subroutine proceeds in phases.  During each phase, \textsc{UpdateSolMain} maintains a data structure $(B, A, \eta_\text{old})$. The set $B = Q_{t'}$ is a fixed base set chosen during the first time step $t'$ of the current phase.  $A$ is a dynamic data structure used by a dynamic submodular maximization algorithm \textsc{Dynamic} that initializes $A$ over function $f_B$, cardinality constraint $k - |B|$, and a parameter $\gamma$ to be later discussed. If a new phase starts at time $t$, note that if $(Q_t, R_t)$ are strongly robust, then the only predicted active elements $V_t^1$ that are needed to find a good solution at time $t$ are, in addition to $B = Q_t$, the small set of elements $R_t$ that are also in $V_t^1$. Thus to find a good solution for the function $f_B$, $(R_t \cap V^1_t) \cup V_t^2$ are inserted into $A$. The solution that \textsc{UpdateSolMain} outputs at time step $t$  are the active elements  $B \cap V_t$ that are in the base for the current phase, together with the solution $\dynamicsol(A)$ maintained by  $A$.

\begin{algorithm}[H]
\setstretch{1.05}
\caption{$\textsc{UpdateSolMain}$}
\textbf{Input:} function $f$, data structure $(B, A, \eta_{\text{old}})$,  constraint $k$,  precomputations $P_t = (Q_t,  R_t)$, 
 $t$, upper bound $\eta'_t$ on prediction error, $V^1_t$, $V^2_t$, $V_{t-1}$, parameter $\gamma_t$ 
\begin{algorithmic}[1]
\If{$t = 1$ or  $|\text{Ops}^{\star}(A)| > \frac{\eta_{\text{old}}}{2}+w $}
\Comment{Start a new phase}
  \State{$B \leftarrow Q_t$} 
  \State{$A \leftarrow \textsc{DynamicInit}(f_{B}, k - |B|,  \gamma = \gamma_t (k - |B|) / k)$}
    \For{$a \in (R_t \cap V^1_t) \cup V_t^2$}  
\State{$\dynamicins(A, a)$}
\EndFor
\State{$\eta_{\text{old}} \leftarrow \eta'_t$}
\Else
\Comment{Continue the current phase}
  \For{$a \in V_{t} \setminus V_{t-1}$}  
\State{$\dynamicins(A, a)$}
\EndFor
  \For{$a \in (\elem(A \cap V_{t-1})\setminus V_{t}$}  
\State{$\dynamicdel(A, a)$ }
\EndFor
\EndIf 
\State{$S \leftarrow (B \cup \dynamicsol(A))\cap V_t$}
\State{\textbf{return} $(B, A, \eta_{\text{old}}), S$} 
\end{algorithmic}
\label{alg:updatesol}
\end{algorithm}


During the next time steps $t$ of the current phase, if an element $a$ is inserted into the stream then $a$ is inserted in $A$ (independently of the predictions). If an element is deleted from the stream, then if it was in $A$, it is deleted from $A$.
We define $\text{Ops}^{\star}(A)$ to be the collection of all insertion and deletion operations to $A$, excluding the insertions of elements in $(R_t \cap V^1_t) \cup V_t^2$ at the time $t$ where $A$ was initialized.
The current phase ends when $\text{Ops}^{\star}(A)$ exceeds $\eta_\text{old}/2 + w.$ Since the update time of the dynamic algorithm in \cite{lattanzi2020fully} depends on length of the stream, we upper bound the length of the stream handled by $A$ during a phase.

\paragraph*{The approximation} The parameter $\gamma_t$ corresponds to a guess for $\OPT_t := \OPT(V_t)$. In Section~\ref{sec:mainresult}, we present the \textsc{UpdateSolFull} subroutine which calls \textsc{UpdateSolMain} with different guesses $\gamma_t$. This parameter $\gamma_t$ is needed when initializing $\textsc{Dynamic}$ because our analysis requires that \textsc{Dynamic}
satisfies a property that we call threshold-based, which we formally define next.


\begin{definition}
\label{def:dynamic_thres} A dynamic algorithm
 $\dynamic$  is  threshold-based  if, when initialized with an arbitrary threshold parameter $\gamma$ such that $\gamma \leq \OPT_t \leq (1+\epsilon)\gamma$, a cardinality constraint $k$, and $\epsilon > 0$, it maintains a data structure $A_t$ and solution $\sol_t = \dynamicsol(A_t)$ that satisfy, for all time steps  $t$,
\begin{enumerate}
    \item $f(\sol_t) \geq \frac{\gamma}{2k} |\sol_t|  $, and 
    \item if $|\sol_t| < (1-\epsilon) k$ then 
    for any set $S \subseteq V(A_t)$ we have $$f_{\sol_t}(S) < \frac{|S| \gamma}{2k} + \epsilon \gamma.$$
\end{enumerate}
\end{definition}

\begin{restatable}{rLem}{dynamicthres}
\label{lem:dynamicthres}
The $\dynamic$ algorithm of Lattanzi et al.~\cite{lattanzi2020fully} 
is threshold-based.
\end{restatable}
Note that the initially published algorithm of \cite{lattanzi2020fully} had an issue with correctness; here we refer to the revised version.
We provide a proof for this lemma in Appendix~\ref{app:updatesol}. The main lemma for the approximation guarantee is the following.



\begin{lemma}
\label{lem:apxpred} Consider the data structure $(B, A, \eta_{\text{old}})$ returned by $\textsc{UpdateSolMain}$ at time $t$. Let $t'$ be the time at which $A$ was initialized,  $(Q_{t'}, R_{t'})$ and  $\gamma_{t'}$ be the  precomputations and guess for $\OPT_{t'}$  inputs to $\textsc{UpdateSolMain}$ at time $t'$.  If  $(Q_{t'}, R_{t'})$ are $(d= 2(\eta_{\text{old}} + 2w), \epsilon, \gamma_{t'})$-strongly robust, $\gamma_{t'}$ is such that $\gamma_{t'} \leq \OPT_t \leq (1+\epsilon) \gamma_{t'}$, and $\dynamic$ is a threshold-based dynamic algorithm, then the set $S$ returned by $\textsc{UpdateSolMain}$ is such that
$$\E[f(S)] \geq \frac{1-5\epsilon}{2}\gamma_{t'}.$$
\end{lemma}




To prove Lemma~\ref{lem:apxpred}, we first bound the value of $Q_{t'} \cap V_t$.
\begin{lemma}
\label{lem:Q_robustness}
Consider the data structure $(B, A, \eta_{\text{old}})$ returned by $\textsc{UpdateSolMain}$ at time $t$. Let $t'$ be the time at which $A$ was initialized,  $(Q_{t'}, R_{t'})$ and  $\gamma_{t'}$ be the  precomputations and guess for $\OPT_{t'}$  inputs to $\textsc{UpdateSolMain}$ at time $t'$. If precomputations $(Q_{t'}, R_{t'})$ are $(d= 2(\eta_{\text{old}} + 2w), \epsilon, \gamma_{t'})$-strongly robust, then we have
$$\E[f(Q_{t'} \cap V_t)] \geq (1-\epsilon)   f(Q_{t'}) \geq (1-\epsilon)   |Q_{t'}|  \gamma_{t'} / (2k).$$
\end{lemma}
\begin{proof}
We  first show that $|Q_{t'} \setminus V_t | \leq d$. 
We know that $Q_{t'} \subseteq \hV_{t'}$.
Firstly, we have $|\hV_{t'} \setminus V_{t'}| \leq \eta_t + 2w \leq \eta_{\text{old}}+2w$ due to the Lemma~\ref{lem:stream_robust}. Next, note that the number of insertions and deletions $\text{Ops}^{\star}(A)$ to $A$ between time $t'$ and $t$ is at most $\frac{\eta_{\text{old}}}{2}+w$, otherwise $A$ would have been reinitialized due to the condition of $\textsc{UpdateSolMain}$ to start a new phase. This implies that $|V_{t'} \setminus V_t |\leq \eta_{\text{old}} + 2w$.
Hence, we have that 
$$|Q_{t'} \setminus V_t| \leq |\hV_{t'} \setminus V_t | \leq  |\hV_{t'} \setminus V_{t'} | + |V_{t'} \setminus V_{t} | \leq 2(\eta_{\text{old}} + 2w) = d.$$
We conclude that $\E[f(Q_{t'} \cap V_t)]\geq (1-\epsilon)   f(Q_{t'}) \geq (1-\epsilon)   |Q_{t'}|  \frac{\gamma_{t'}}{2k}$, where the first inequality is by the robustness property of strongly-robust precomputations and the second by their value property. 
\end{proof}

\textit{Proof of Lemma~\ref{lem:apxpred}.} There are three cases.

\begin{enumerate}
\item  $|Q_{t'}| \geq k$. We have that $\E[f(S)] \geq \E[f(Q_{t'} \cap V_t)] \geq (1-\epsilon)   |Q_{t'}|  \frac{\gamma_{t'}}{2k} \geq (1-\epsilon)    \frac{\gamma_{t'}}{2}$,
where the first inequality is by monotonicity and since $S \subseteq B \cap V_t = Q_{t'} \cap V_t$, the second by Lemma~\ref{lem:Q_robustness}, and the last by  the condition for this first case.

\item $|\dynamicsol(A)| \geq (1-\epsilon)(k- |Q_{t'}|)$. Let $\sol_t = \dynamicsol(A)$ and recall that  $A$ was initialized with $\dynamicinit$ over function $f_{Q_{t'}}$ with cardinality constraint $k - |B|$ and threshold parameter $\gamma_{t'} (k - |B|)/k$. We get
\begin{align*} \E[f(S)] &=  \E[f(Q_{t'}\cap V_t) + f_{Q_{t'}\cap V_t}(\sol_t)] & \text{definition of $S$} \\
 &\geq \E[f(Q_{t'}\cap V_t) + f_{Q_{t'}}(\sol_t)]  & \text{submodularity} \ \\
&\geq (1-\epsilon)   |Q_{t'}|  \frac{\gamma_{t'}}{2k} + \E[f_{Q_{t'}}(\sol_t)] & \text{Lemma~\ref{lem:Q_robustness}} \\
& \geq (1-\epsilon)   |Q_{t'}| \frac{\gamma_{t'}}{2k}  + \frac{\gamma_{t'} (k - |B|)/k}{2(k-|B|)} |\sol_t| & \text{Definition~\ref{def:dynamic_thres}} \\
& \geq (1-\epsilon)   |Q_{t'}|  \frac{\gamma_{t'}}{2k} + \frac{\gamma_{t'}}{2k}  (1-\epsilon)(k- |Q_{t'}|) & \text{case assumption} \\
&\geq \frac{(1-\epsilon) \gamma_{t'}}{2}. & 
\end{align*}

\item $|Q_{t'}| < k$ and $|\sol_t| < (1-\epsilon)(k- |Q_{t'}|)$: 
Recall that $R_{t'} \subseteq \hV_{t'}$. Also, let $\bar{R}_t = (V_t \cap R_{t'}) \cup (V_t \setminus \hV_{t'})$.
In this case, we have that
\begin{align*}
  f(O_t) & \leq_{(a)} \E[f(O_t \cup \sol_t \cup Q_{t'})] \\
& = \E[f(\sol_t \cup Q_{t'})] +  \E[f_{\sol_t \cup Q_{t'}}(O_t\setminus \bar{R}_t)] +  \E[f_{\sol_t \cup Q_{t'} \cup (O_t\setminus \bar{R}_t)}(O_t\cap \bar{R}_t)] \\
& \leq_{(b)} \E[f(\sol_t \cup Q_{t'})] +  \E[f_{Q_{t'}}(O_t\setminus \bar{R}_t)] + \E[f_{\sol_t}(O_t\cap \bar{R}_t)] \\
& \leq_{(c)} \E[f(\sol_t \cup Q_{t'})] + |O_t\setminus \bar{R}_t| \cdot \frac{\gamma_{t'}}{2k} + \epsilon \gamma_{t'} + \E[f_{\sol_t}(O_t\cap \bar{R}_t)] \\
& \leq_{(d)} \E[f(\sol_t \cup Q_{t'})] + |O_t\setminus \bar{R}_t| \cdot \frac{\gamma_{t'}}{2k} + \epsilon \gamma_{t'} + |O_t\cap\bar{R}_t| \cdot \frac{\gamma_{t'}}{2k} + \epsilon \gamma_{t'} \\
&\leq_{(e)} \E[f(\sol_t \cup Q_{t'})] +  (1+4\epsilon) \frac{\gamma_{t'}}{2} \,.
\end{align*}
where $(a)$ is by monotonicity,
$(b)$ is by submodularity, $(c)$ is since $(Q_{t'}, R_{t'})$ are $(d= 2(\eta_{\text{old}} + 2w), \epsilon, \gamma_{t'})$-strongly robust (value property) and since we have that  $O_t \setminus  \bar{R}_t = O_t \setminus ((V_t \cap R_{t'} ) \cup (V_t \setminus \hV_{t'})) \subseteq \hV_{t'} \setminus R_{t'}$ and $Q_{t'} < k$, $(d)$ is since $\dynamic$ is threshold-based, $|\sol_t| < (1-\epsilon)(k- |Q_{t'}|)$, and $\bar{R}_t \subseteq V(A)$, and $(e)$ is since $|O_t| \leq k$. The above series of inequalities implies that $$\E[f(\sol_t \cup Q_{t'})] \geq f(O_t) - \frac{(1+4\epsilon) \gamma_{t'}}{2} = \frac{(1-4\epsilon) \gamma_{t'}}{2}.$$ We conclude that 
\begin{align*}
 \E[f(S)] & = \E[f((Q_{t'} \cup \sol_t) \cap V_t )] &  \text{definition of $S$} \\
 & = \E[f((Q_{t'}\cap V_t) \cup \sol_t  )] & \text{$\sol_t \subseteq V_t$} \\
&\geq  \E[f(Q_{t'}\cap V_t)] + \E[f_{Q_{t'}}(\sol_t)] & \text{submodularity} \\
& \geq  (1-\epsilon) f(Q_{t'}) + \E[f_{Q_{t'}}(\sol_t)]  & \text{Lemma~\ref{lem:Q_robustness}}\\
  &\geq  (1-\epsilon) \cdot \E[f(\sol_t \cup Q_{t'})] &\\
  & \geq \frac{(1-5\epsilon) \gamma_{t'}}{2} & \qed
\end{align*}
\end{enumerate}

\paragraph*{The update time} We next analyze the query complexity of \textsc{UpdateSolMain}.


\begin{restatable}{rLem}{utpred}
\label{lem:utpred} Consider the data structure $(B, A, \heta)$ returned by $\textsc{UpdateSolMain}$ at time $t$. Let $t'$ be the time 
at which $A$ was initialized and $(Q_{t'}, R_{t'})$  be the  precomputations  input to $\textsc{UpdateSolMain}$ at time $t'$. If precomputations $(Q_{t'}, R_{t'})$ are $(d= 2(\heta + 2w), \epsilon, \gamma)$ strongly-robust, 
then the total number of queries performed by $\updatesol$ during the $t - t'$ time steps between time $t'$ and time $t$ is 
$O(u((\heta + w  + k)\log k, k) \cdot (\heta + w  + k)\log k) $.
Additionally, the number of queries between time $1$ and $t$
is upper bounded by $O(u(t,k) \cdot t) $.
\end{restatable}
\begin{proof}
The only queries made by 
$\updatesol$ are due to calls to $\dynamic$. 
Hence, we calculate the total number of operations $\text{Ops}(A)$ between time $t'$ and $t$.
The number of insertions 
at time $t'$ is
\begin{align*}
|R_{t'} \cap V^1_{t'}  | + |V_{t'}^2 | & =_{(1)} O((\heta + w  + k)\log k) + |V_{t'}^2 | \\
& \leq_{(2)} O((\heta + w  + k)\log k) + \heta \\
& = O((\heta + w  + k)\log k)),
\end{align*}
where $(1)$ is  by the size property of the $(d= 2(\heta + 2w), \epsilon, \gamma)$ strongly-robust precomputations and $(2)$ is since $|V_{t'}^2| \leq \eta_{t'} \leq \eta'_{t'} = \eta_{\text{old}}$.

Moreover, the total number of operations $\text{Ops}^*(A)$ between time $t'+1$ and $t$
is at most $\heta/2 + w = d/4$,
otherwise $A$ would have been reinitialized due to the condition of $\textsc{UpdateSolMain}$ to start a new phase.
Hence, the total query complexity between time $t'$ and time $t$ is at most $u(O((\heta + w  + k)\log k), k- |Q_{t'}|)\cdot O((\heta + w  + k)\log k) = O(u((\heta + w  + k)\log k, k) \cdot (\heta + w  + k)\log k)$.

In the case $t'=1$, the number of operations at time $t'$ is $1$ since $|V_{t'}| = 1$,
and the number of operations from time $2$ to $t$ is at most $t-1$. 
This gives the bound of $O(u(t,k)\cdot t)$ when $t'=1$. 
\end{proof}

%% file: 5.precomputations.tex
In this section, we provide a \textsc{Precomputations} subroutine that has an improved query complexity compared to the warm-up precomputations subroutine. Recall that the warm-up  subroutine computes a robust solution over predicted elements $\hat{V}_t$, independently for all times $t$. The improved \textsc{Precomputations} does not do this independently for each time step. Instead, it relies on the following lemma that shows that the data structure maintained by the dynamic algorithm of \cite{lattanzi2020fully} can be used to find a strongly robust solution without any additional query.

\begin{restatable}{rLem}{robustfromdynamic}
\label{lem:robust_from_dynamic} 
Let $\dynamic(\gamma, \epsilon)$ be the dynamic submodular maximization algorithm of 
\cite{lattanzi2020fully} and  $A$ be the data structure it maintains.
There is a $\textsc{Robust1FromDynamic}$ 
algorithm such that, given as input a  deletion size parameter $d$,
and the data structure $A$ at time $t$ with $\gamma \leq \OPT_t \leq (1+\epsilon)\gamma$, it outputs  sets $(Q,R)$ that are $(d, \epsilon, \gamma)$-strongly robust 
with respect to the ground set $V_t$.
Moreover, this algorithm does not perform any oracle queries.
\end{restatable}

This lemma is proved in Appendix~\ref{app:precomputation}. The improved \textsc{PrecomputationsMain} subroutine runs the dynamic algorithm of \cite{lattanzi2020fully} and then, using the $\textsc{Robust1FromDynamic}$ algorithm of Lemma~\ref{lem:robust_from_dynamic}, computes a strongly-robust set from the data structure maintained by the dynamic algorithm.

\begin{algorithm}[H]
\caption{\textsc{PrecomputationsMain}}
\setstretch{1.05}
\hspace*{\algorithmicindent} \textbf{Input:} function $f : 2^V \> \R$,  constraint $k$,  predicted elements $\hat{V}_1, \ldots, \hat{V}_{n} \subseteq V$, 
time error tolerance $w$, parameter $\gamma$, 
parameter $h$
\begin{algorithmic}[1]
\State $\hat{A} \leftarrow  \textsc{DynamicInit}(f, k, \gamma)$ 
\For{$t=1$ to $n$}
\For{$a \in\hV_{t} \setminus \hV_{t-1} $}
\State{$\dynamicins(\hat{A}, a)$}
\EndFor
\For{ $a  \in V(\hat{A}) \setminus \hV_{t}$}  
\State{$\dynamicdel(\hat{A}, a)$}
\EndFor
\If{$|V(\hat{A})| > 0$}
\State{$ Q_t, R_{t} \leftarrow \textsc{Robust1FromDynamic}(f, \hat{A}, k, 2(h + 2w))$} 
\EndIf
\EndFor
\State{\textbf{return} $\{(Q_t, R_t)\}_{t=1}^n$}
\end{algorithmic}
\end{algorithm}

The parameters $\gamma$ and $h$ correspond to guesses for $\OPT$ and $\eta$ respectively. In In Section~\ref{sec:mainresult}, we present the \textsc{PrecomputationsFull} subroutine which calls \textsc{PrecomputationsMain} with different guesses $\gamma$ and $h$.

\begin{lemma}
\label{lem:precomputation2}
The total query complexity of the \textsc{PrecomputationsMain} algorithm is $n \cdot u(n, k)$, where 
$u(n,k)$ is the amortized query complexity of calls to $\dynamic$.
\end{lemma}
\begin{proof}
It is easy to observe that the algorithm makes at most $n$ calls to 
$\dynamicins$ or $\dynamicdel$ since $\sum_t |\hV_{t} \setminus \hV_{t-1}| + | \hV_{t-1} \setminus \hV_{t} | \leq n$.
Hence, the total query complexity due to calls to  $\dynamic$ is $n \cdot u(n,k)$.
Moreover, the calls to $\textsc{Robust1FromDynamic}$ do not incur any additional queries due to Lemma~\ref{lem:robust_from_dynamic}. 
Hence, the total number of queries performed in the precomputation phase is given by $n \cdot u(n,k)$.
\end{proof}

%% file: 6.main_result.tex
The \textsc{UpdateSolMain} and \textsc{PrecomputationsMain} subroutines use guesses $\gamma_t$ and $h$ for the optimal value $\OPT_t$ at time $t$ and the total prediction error $\eta$. In this section, we describe the full \textsc{UpdateSolFull} and \textsc{PrecomputationsFull} subroutines that call the previously defined subroutines over different guesses $\gamma_t$ and $h$. The parameters of these calls must be carefully designed to bound the the streaming amortized query complexity and precomputations query complexity.

\subsection{The full \textsc{Precomputations} subroutine}
\label{app:fullprecomp}

We define $H = \{n / 2^i : i \in \{\log_2(\max\{\frac{n}{k-2w}, 1\}), \cdots, \log_2(n) -1, \log_2 n \}\}$ to be a set of guesses for  the prediction error $\eta$. 
Since \textsc{PrecomputationsMain} requires a guess $h$ for  $\eta$, \textsc{PrecomputationsFull} calls \textsc{PrecomputationsMain} over all guesses $h \in H$. 
We ensure that the minimum guess $h$ is such that  $h+2w$ is at least $k$.
The challenge with the guess $\gamma$ of $\OPT$ needed for \textsc{PrecomputationsMain} is that $\OPT$ can have arbitrarily large changes between two time steps. Thus, with $\hat{V}_1, \ldots, \hat{V}_n \subseteq V$, we define the collection of guesses for $\OPT$ to be $\Gamma = \{(1+ \epsilon)^0\min_{a \in V} f(a), (1+ \epsilon)^1\min_{a \in V} f(a), \ldots, f(V)\}$, which can be an arbitrarily large set.  

Instead of having a bounded number of guesses $\gamma$, we consider a subset  $\hat{V}_t(\gamma) \subseteq \hat{V}_t$ of the predicted elements at time $t$ for each guess $\gamma \in \Gamma$ such that for each element $a$ is, there is a bounded number of guesses $\gamma \in \Gamma$ such that $a \in \hat{V}_t(\gamma)$. More precisely, for any set $T$, we define $T(\gamma) := \{a \in T: \frac{\epsilon \gamma}{k}  \leq f(a) \leq 2\gamma\}$ and
$\Gamma(a) = \{\gamma \in \Gamma: f(a) \leq \gamma \leq \epsilon^{-1} k f(a)\}.$ The \textsc{PrecomputationsFull} subroutine outputs, for every time step $t$, strongly-robust sets $(Q^{\gamma, h}_t, R^{\gamma, h}_t)$, for all guesses $\gamma \in \Gamma$ and $h \in H$. If  $\hat{V}_t(\gamma) = \emptyset$, we assume that $Q^{\gamma, h}_t = \emptyset$ and $R^{\gamma, h}_t = \emptyset.$

\begin{algorithm}[H]
\caption{\textsc{PrecomputationsFull}}
\setstretch{1.05}
\hspace*{\algorithmicindent} \textbf{Input:} function $f$,  constraint $k$,  predicted active elements $\hat{V}_1, \ldots, \hat{V}_{n} \subseteq V$,  time error tolerance $w$
\begin{algorithmic}[1]
\For{$\gamma \in \Gamma$ such that $|\hat{V}_t(\gamma)| > 0$ for some $t \in [n]$}
\For{$h \in H$}
\State{$\{(Q^{\gamma, h}_t, R^{\gamma, h}_t)\}_{t=1}^n \leftarrow \textsc{PrecomputationsMain}(\hat{V}_1(\gamma), \ldots, \hat{V}_{n}(\gamma), \gamma, h)$}
\EndFor
\EndFor
\State{\textbf{return} $\{\{(Q^{\gamma, h}_t, R^{\gamma, h}_t)\}_{\gamma \in \Gamma, h \in H}\}_{t=1}^n$}
\end{algorithmic}
\label{alg:tgreedy_wu}
\end{algorithm}

\begin{lemma}
\label{lem:precompqc}
The total query complexity of \textsc{PrecomputationsFull} is $$O(n \cdot \log(n) \cdot  \log(k) \cdot u(n,k)).$$
\end{lemma}
\begin{proof}
For any $\gamma \in \Gamma$, let $n^{\gamma}$ be the length of the stream corresponding to predicted active elements $\hat{V}_1(\gamma), \ldots, \hat{V}_{n}(\gamma)$. By Lemma~\ref{lem:precomputation2}, the total query complexity is 
\begin{align*}
    \sum_{\gamma \in \Gamma} |H| n^{\gamma} u(n^{\gamma}, k)  &\leq \sum_{\gamma \in \Gamma} \log(n) \cdot 2|\hat{V}(\gamma)|   \cdot u(n, k)  \\
    &= 2 (\log n)  u(n, k)  \sum_{a \in \hat{V}}  |\Gamma(a)|    \\
    & = 2 (\log n)  u(n, k)  n  (\log k). \qedhere
\end{align*}
\end{proof}

\subsection{The full \textsc{UpdateSol} subroutine}
\label{app:fullupdatesol}

$\textsc{UpdateSolFull}$ takes as input the  strongly robust sets $(Q^{\gamma, h}_t, R^{\gamma, h}_t)$, for all guesses $\gamma \in \Gamma$ and $h \in H$. It maintains a data structure $\{(B^\gamma, A^{\gamma}, \eta^{\gamma})\}_{\gamma \in \Gamma}$ where $(B^\gamma, A^{\gamma}, \eta^{\gamma})$ is the data structure maintained by \textsc{UpdateSolMain} over guess $\gamma$ for $\OPT$. \textsc{UpdateSolMain} is called over all guesses $\gamma \in \Gamma$ such that there is at least one active element $ a \in V_t(\gamma)$. The precomputations given as input to \textsc{UpdateSolMain} with guess $\gamma \in \Gamma$ are $(Q^{\gamma,\eta'_t}_t,  R^{\gamma,\eta'_t}_t)$ where $\eta'_t \in H$ is the closest guess  in $H$ to the current prediction error $\eta_t$ that has value at least $\eta_t$. 
Note that $\eta'_t$ is such that $\eta'_t + 2w \geq k$ due to definition of $H$.
The solution returned at time $t$ by $\textsc{UpdateSolFull}$ is the best solution found by \textsc{UpdateSolMain} over all guesses $\gamma \in \Gamma$.

\begin{algorithm}[H]
\setstretch{1.05}
\caption{$\textsc{UpdateSolFull}$}
\hspace*{\algorithmicindent} \textbf{Input:} function $f : 2^V \> \R$, data structure $A$,  constraint $k$,  precomputations $P_t = \{(Q^{\gamma, \eta}_t,  R^{\gamma, \eta}_t)\}_{\gamma \in \Gamma, \eta \in H}$, time $t$, current prediction error $\eta_t$,  $V^1_t$, $V^2_t$, $\hV_t$
\begin{algorithmic}[1]
\State{$\{A^{\gamma}\}_{\gamma \in \Gamma} \leftarrow A$}
\State{$\eta'_t \leftarrow \min\{h \in H: h \geq \eta_t\}$}
\For{$\gamma \in \Gamma$ such that $|V_t(\gamma)| > 0$}
\State{$A^{\gamma}, S^{\gamma} \leftarrow \textsc{UpdateSolMain}(f, A^{\gamma}, k, (Q^{\gamma,\eta'_t}_t,  R^{\gamma,\eta'_t}_t),  t, \eta'_t, V^1_t(\gamma), V^2_t(\gamma), V_{t-1}(\gamma), \gamma)$}
\EndFor
\State{$A \leftarrow \{A^{\gamma}\}_{\gamma \in \Gamma}$}
\State{$S \leftarrow \argmax_{\gamma \in \Gamma : |V_t(\gamma)| > 0}f(S^{\gamma}) $}
\State{\textbf{return} $A, S$} 
\end{algorithmic}
\label{alg:updatesolfull}
\end{algorithm}

\begin{lemma}
\label{lem:utfpred}
The $\textsc{UpdateSolFull}$ 
algorithm has, over the $n$ time-steps of the streaming phase, an amortized query complexity of  
$O\left(\log^2(k) \cdot u(\eta + w + k, k) \right)$. 
\end{lemma}
\begin{proof}

For every $\gamma \in \Gamma$, we consider the following stream associated to $\gamma$:
$\{(a'_t, o'_t)\}_{t=1}^n$ 
where $
(a'_t, o'_t) =  (a_t, o_t) $ if  $\gamma \in \Gamma(a_t)$ and $(a'_t, o'_t) = (\NULL, \NULL)$ otherwise. $\NULL$ is a dummy symbol denoting an 
absence of an insertion or deletion.
Let $n^\gamma$ be the number of insertion/deletions in this new stream, i.e.,
$n^\gamma = \sum_{t=1}^n \1[a_t \neq \NULL]$.

Using Lemma~\ref{lem:utpred} the total query complexity 
due to a single phase corresponding to a $(\gamma, \eta^{\gamma})$ pair is $O(u((\eta^{\gamma} + w  + k)\log k, k) \cdot (\eta^{\gamma} + w  + k)\log k)$.
Next, note that there is a one-to-one mapping between  insertions/deletions to $A^\gamma$ at line 9 or line 11 of \textsc{UpdateSol} and elements of  the stream $\{(a'_t, o'_t)\}_{t = 1}^{n}$. Thus, since a phase ends when $|\text{Ops}^{\star}(A)| > \eta^{\gamma} / 2 + w$, all phases except the last phase process at least $\eta^{\gamma} / 2 + w$ non-null elements from 
the stream $\{(a'_t, o'_t)\}_{t = 1}^{n}$.
Thus, if there are at least two phases corresponding to $n^\gamma$ then the total query complexity over these $\eta^{\gamma} / 2 + w$ non-null elements is $O(u(\eta^{\gamma}  + w + k, k) \cdot  (\eta^{\gamma} + w  + k)\log k \cdot \frac{n^\gamma}{\eta^{\gamma}/2 + w})$. 
Since, $\eta^{\gamma}/2 + w > k$ we have that the
amortized query complexity over $n^{\gamma}$ non-null
elements is $O(u(\eta^{\gamma}  + w + k, k) \cdot  \log k)$.
We also have that $O(u(\eta^{\gamma}  + w + k, k) \cdot  \log k) = O(u(\eta + w + k, k)\cdot \log k)$ because either $\eta + 2w > k$ which implies that $\eta^{\gamma} + w = O(\eta+w)$ or
$\eta + 2w < k$ in which case  $\eta^{\gamma}  + w + k = O(k)$.

If there is only one phase, i.e.\ when $n^\gamma \leq \eta^\gamma/2 + w$, then the amortized query complexity for the first phase is upper bounded by $u(n^\gamma, k) = O(u(\eta^{\gamma} / 2 + w, k))$. 
Thus, the amortized query complexity due to $\gamma$ is $O(u(\eta + w + k, k) \log k)$, for any $\gamma \in \Gamma$. 
We get that the total query complexity is
\begin{align*}
\sum_{\gamma \in \Gamma} O(u(\eta+ w+k,k) \cdot n^\gamma \cdot \log k) &=
\sum_{\gamma \in \Gamma} \sum_{a \in \Gamma(a)}  O(u(\eta + w+k,k) \log k)  \\
&=
\sum_{a \in V} \sum_{\gamma \in \Gamma(a)}  O(u(\eta + w + k,k) \log k)\\
& =   O(n \log^2 k) \cdot u(\eta + w + k,k) \qedhere
\end{align*}
\end{proof}

\subsection{The main result}

By combining the algorithmic framework (Algorithm~\ref{alg:framework}) together with subroutines \textsc{UpdateSolFull} and \textsc{PrecomputationsFull}, we obtain our main result.

\begin{restatable}{rThm}{thmmain}
\label{thm:main}
Algorithm~\ref{alg:framework} with subroutines \textsc{UpdateSolFull} and \textsc{PrecomputationsFull}   is a dynamic algorithm with predictions that, for any tolerance $w$ and constant $\epsilon > 0$, achieves  an amortized query complexity during the streaming phase of $O(\poly(\log \eta, \log w, \log k))$,  
 an approximation of $1/2 - \epsilon$ in expectation, and a query complexity of $\tilde{O}(n)$ during the precomputation phase. 
\end{restatable}

\begin{proof} The dynamic algorithm \textsc{Dynamic} used by the \textsc{Precomputations} and \textsc{UpdateSol} subroutines is the algorithm of Lattanzi et al.~\cite{lattanzi2020fully} with amortized update time $u(n, k) = O(\poly(\log n, \log k))$. By Lemma~\ref{lem:utfpred}, the amortized query complexity is $$O\left(\log^2(k) \cdot u(\eta + 2w + k, k) \right) = O\left(\poly(\log(\eta + w + k),  \log k) \right).$$ 

For the approximation, consider some arbitrary time $t$. Let $\gamma^{\star} = \max\{\gamma \in \Gamma : \gamma \leq (1-\epsilon)\OPT_t  \}$.  Let $\OPT'_t := \max_{S \subseteq V_t(\gamma^{\star}) : |S| \leq k} f(S)$. We have that
$$\OPT'_t  \geq f(O_t \cap V_t(\gamma^{\star})) \geq_{(1)} f(O_t) - \sum_{o \in O_t \cap V_t(\gamma^{\star})} f(o) \geq_{(2)} f(O_t) - k \cdot  \frac{\epsilon \gamma^{\star}}{k} \geq_{(3)} (1 - \epsilon) \OPT_t $$
where $(1)$ is by submodularity, $(2)$ is by definition of $V_t(\gamma^{\star})$, and $(3)$ by definition of $\gamma^{\star}$. Consider the calls to $\textsc{UpdateSol}$ by $\textsc{UpdateSolFull}$ with $\gamma = \gamma^{\star}$. Let $t'$ be time at which $A_t^{\gamma^{\star}}$ was initialized by  $\textsc{UpdateSol}$ with precomputations $(Q_{t'}, R_{t'}) = (Q^{\gamma^{\star},\eta'_{t'}}_{t'},  R^{\gamma^{\star},\eta'_{t'}}_{t'})$, where this equality is by definition of $\textsc{UpdateSolFull}$. 

Next, we show that the conditions to apply Lemma~\ref{lem:apxpred} are satisfied.
By definition of $\textsc{PrecomputationsFull}$ and Lemma~\ref{lem:robust_from_dynamic}, $(Q_{t'}, R_{t'})$ are $(d = 2(\eta'_{t'} + 2w), \epsilon, \gamma^{\star})$-strongly robust with respect to $V_t(\gamma^{\star})$ with $\gamma^{\star} = \gamma_{t'}$. Since $\eta'_{t'} \geq \eta_{t'},$ we have that $(Q_{t'}, R_{t'})$ are $(d = 2(\eta_{\text{old}} + 2w), \epsilon, \gamma^{\star})$. We also have that 
$$\gamma_{t'} \leq (1-\epsilon) \OPT_t \leq \OPT'_t \leq \OPT_t \leq  (1+ \epsilon)\gamma_{t'} / (1-\epsilon).$$ Thus, with $\epsilon' > 0$ such that $(1 + \epsilon') = (1+ \epsilon) / (1-\epsilon)$, we get $\gamma_{t'} \leq \OPT'_t\leq  (1+ \epsilon')\gamma_{t'}.$ By Lemma~\ref{lem:dynamicthres}, \textsc{Dynamic} is a threshold-based algorithm. Thus, all the conditions of Lemma~\ref{lem:apxpred} are satisfied and we get that the solution $S^{\gamma^{\star}}_t$ returned by the call to $\textsc{UpdateSol}$ with $\gamma = \gamma^{\star}$ at time $t$, where $A_t^{\gamma^{\star}}$ was initialized by  $\textsc{UpdateSol}$ at time $t'$ with $\gamma = \gamma^{\star}$, is such that $\E[f(S^{\gamma^{\star}}_t)] \geq \frac{1- 5 \epsilon'}{2} \gamma^{\star}  \geq (1- \epsilon)\frac{1- 5 \epsilon'}{2} \OPT_t.$ Finally, since $\textsc{UpdateSolFull}$ returns, among all the solutions returned by $\textsc{UpdateSol}$, the one with highest value, it returns a set $S_t$ such that $\E[f(S_t)] \geq \E[f(S^{\gamma^{\star}}_t)] \geq (1- \epsilon)\frac{1- 5 \epsilon'}{2} \OPT_t.$  Finally, the precomputation query complexity is by Lemma~\ref{lem:precompqc} with $u(n, k) = O(\poly(\log n, \log k)$.
\end{proof}

%% file: appendix_warmup.tex
\lemstreamrobust*
\begin{proof}
Let $E = \{a \in V: |\hat{t}^+_a - t^+_a| > w \text{ or } |\hat{t}^-_a - t^-_a| > w\}$.
Hence, for all $a \not \in E$, we have $|\hat{t}^+_a - t^+_a| \leq w \text{ and } |\hat{t}^-_a - t^-_a| \leq w$.
By Definition~\ref{def:error}, we have that $|E| = \eta$.
We first note that any $a \in \hV_t \setminus E$ such that 
$\hat{t}_a^+ \leq t - w$  and $\hat{t}_a^- > t+w$ also belongs to $V_t$. This is because $a \not\in E$ implies that $t_a^+ \leq t$ and $t_a^- > t$.
Hence, the only elements that are present in $\hV_t \setminus E$ but are absent from $V_t$ are such that
$\hat{t}_a^+ > t - w$  or $\hat{t}_a^- \leq t+w$.
Since, the only elements in $\hat{V}_{t}$ are such that  $\hat{t}^+_a \leq t+ w \text{ and } \hat{t}^-_a \geq t - w$, this implies that there can be at most $2w$ such elements. 
Combining with the fact that $|E| = \eta$,
we get $|\hV_t \setminus V_t | \leq |E| + |(\hV_t \setminus E) \setminus V_t| \leq \eta + 2w$.

\end{proof}

\robustwu*
\begin{proof}
We first show that  at every time step $t$,  $f(S^1)  \geq \alpha_1  \cdot \OPT(V_t^1)$. Fix some time step $t$. First, note that
$$P_t \setminus (\hat{V}_t \setminus V^1_t) = (P_t \setminus \hat{V}_t) \cup (P_t \cap V^1_t) =  P_t \cap V^1_t$$
where the second equality is since $P_t \subseteq \hV_t$. In addition, we have that $|\hat{V}_t \setminus V^1_t| = |\hat{V}_t \setminus V_t| \leq \eta + 2w = d$, where the inequality is by Lemma~\ref{lem:stream_robust}.  With $D = \hat{V}_t \setminus V^1_t$ such that $|D| \leq d$, $R = P_t$, and $N = \hV_t$, we thus have by Definition~\ref{def:robust} that the output $S^1$ of $\robusttwo(f, P_t, \hat{V}_t \setminus V^1_t, k) = \robusttwo(f, R, D, k)$ is such that $S^1 \subseteq P_t \cap V^1_t$ and
$$ \E[f(S^1)] \geq 
\alpha_1  \cdot \max_{\substack{S \subseteq \hat{V}_t \setminus D: \\ |S| \leq k}} f(S) 
= \alpha_1 \cdot  \max_{\substack{S \subseteq V^1_t: \\ |S| \leq k}} f(S) = \alpha_1 \cdot   \OPT(V^1_t).$$

Next, we show that at every time step $t$,  $f(S^2)  \geq \alpha_2 \cdot \OPT(V^2_t)$. Observe that the algorithm runs the dynamic streaming algorithm \dynamic \ over the sequence of ground sets $V_0^u, \ldots, V_n^u$.
By Definition~\ref{def:dynamic}, the solutions returned by \dynamic \ at each time step achieve a  $\alpha_2$ approximation. Since these solutions are $S^u_1, \ldots, S^u_n$, we get that for every time step $t$, $  f(S_t^u) \geq \alpha_2 \max_{S \subseteq V^u_t: |S| \leq k} f(S).$
\end{proof}

\begin{lemma}
\label{lem:partition_approx_wu}
    For any monotone submodular function $f: 2^{N_1 \cup N_2} \rightarrow \R$, if $S_1$ and $S_2$ are such that $f(S_1) \geq \alpha_1 \cdot \max_{S \subseteq N_1: |S| \leq k} f(S)$ and $f(S_2) \geq \alpha_2 \cdot \max_{S \subseteq N_2: |S| \leq k} f(S)$, then $\max\{ f(S_1), f(S_2)\} \geq \frac{1}{2} \cdot\min\{\alpha_1, \alpha_2\} \cdot \max_{S \subseteq N_1 \cup N_2: |S| \leq k} f(S)$.    
\end{lemma}
\begin{proof}
Let $O  = \{o_1, \ldots, o_k\}= \argmax_{S \subseteq N_1 \cup N_2: |S| \leq k} f(S)$ be an arbitrary ordering of the optimal elements. We have that 
\begin{align*}
 \max_{S \subseteq N_1 \cup N_2: |S| \leq k} f(S)  = \ & f(O) \\
 = \ &\sum_{i=1}^k f_{\{o_1, \cdots, o_{i-1}\}}(o_i) \\
	\leq \ &\sum_{i=1}^k \1[o_i \in N_1] f_{\{o_1, \cdots, o_{i-1}\}\cap N_1}(o_i) +  \sum_{i=1}^k \1[o_i \in N_2] f_{\{o_1, \cdots, o_{i-1}\}\cap N_2}(o_i)	\\	
	= \ &f(O \cap N_1) + f(O \cap N_2) \\
  \leq \ &\max_{S \subseteq N_1: |S| \leq k} f(S) + \max_{S \subseteq N_2: |S| \leq k} f(S) \\
  \leq \ &f(S_1) / \alpha_1 + f(S_2) / \alpha_2 \\
  \leq \ &2 \max\{ f(S_1), f(S_2)\} / \min\{\alpha_1, \alpha_2\}
\end{align*}
where the first inequality is by submodularity.
\end{proof}

By combining Lemma~\ref{lem:robust_wu} and Lemma~\ref{lem:partition_approx_wu}, we obtain the main lemma for the approximation guarantee.

\combinedsol*

\sizeA*
\begin{proof}
The number of calls to  $\dynamicins$ is $|\{a : \exists t \text{ s.t. } a \in V_t^2 \setminus V_{t-1}^2\}|$. Since $V_t^2 \setminus V_{t-1}^2 = (V_t \setminus \hV_t)\setminus (V_{t-1} \setminus \hV_{t-1}^1)$,
\begin{align*}
   |\{a : \exists t \text{ s.t. } a \in V_t^2 \setminus V_{t-1}^2\}| & \leq |\{a : \exists t \text{ s.t. } a \in (V_t \setminus V_{t-1})\setminus \hV_t\}| \\
   & \qquad \qquad +  |\{a : \exists t \text{ s.t. } a \in (\hV_{t-1} \setminus \hV_t)\cap V_t\}|.
\end{align*}
Next, we have
\begin{align*}
    |\{a : \exists t \text{ s.t. } a \in (V_t \setminus V_{t-1})\setminus \hV_t\}| 
    & = |\{a : a \not \in \hV_{t^+_a}\}| \\
    & = |\{a : \hat{t}^+_a > t^+_a + w \text{ or } \hat{t}^-_a < t^+_a - w\}| \\
     & = |\{a : \hat{t}^+_a > t^+_a + w \text{ or } \hat{t}^-_a < t^-_a - w\}| \\
    & \leq  |\{a : |\hat{t}^+_a - t^+_a| > w \text{ or } |\hat{t}^-_a - t^-_a| > w \}|\\
    & \leq \eta 
\end{align*}
Similarly,
\begin{align*}
    |\{a : \exists t \text{ s.t. } a \in (\hV_{t-1} \setminus \hV_t)\cap V_t\}| 
    & = |\{a : a \in V_{\hat{t}^-_a + w}\}| \\
    & = |\{a : t^-_a \geq \hat{t}^-_a + w\}| \\
    & \leq  |\{a : |\hat{t}^+_a - t^+_a| > w \text{ or } |\hat{t}^-_a - t^-_a| > w \}|\\
    & \leq \eta 
\end{align*}

Combining the above inequalities, we get that the number of calls to $\dynamicins$ is  $|\{a : \exists t \text{ s.t. } a \in V_t^2 \setminus V_{t-1}^2\}| \leq 2 \eta.$
\end{proof}

%% file: appendix_updatetime.tex
\section{Missing Proof from Section~\ref{sec:updatesol}}
\label{app:updatesol}

\dynamicthres*

\begin{proof}
The first condition that $f(\sol_t) \geq \frac{\gamma}{2k} |\sol_t|$ follows directly from the fact that each element $e$ that is added to $\sol_t$
has a marginal contribution of at least $\gamma/2k$ to the partial solution.
Some elements could have been deleted 
from the partial solution since the time $e$ 
was added, but this can only increase the contribution of $e$ due to submodularity.

We will now show that if $|\sol_t| < (1-\epsilon) k$ then for any 
$S \subseteq V(A_t)$ we have $f_{\sol_t}(S) < \frac{|S| \gamma}{2k} + \epsilon \gamma$.
We will use the set $X$ defined 
in the proof of Theorem 5.1 in \cite{lattanzi2020fully}.
The proof of Theorem 5.1 in \cite{lattanzi2020fully} shows that $|\sol_t| \geq (1-\epsilon)|X|$.
Using this, the condition that $|\sol_t| < (1-\epsilon) k$ implies that $|X| < k$.
Hence, we fall in the case 2 of the 
proof of Theorem 5.1 
which shows that $f(e | X) \leq \gamma/2k $
for all $e \in V_t \setminus \sol_t$.
We then have that 
\begin{align*}
f_{\sol_t}(S) &= f(S \cup \sol_t) - f(\sol_t) \leq f(S \cup X) - f(\sol_t) \\
&= f(S \cup X) - f(X) + f(X) - f(\sol_t) \\ 
&\leq f(X\cup S) - f(X) + f(X) - (1-\epsilon) f(X) \\ 
&\leq \sum_{e \in S} f_X(e)  + \epsilon f(X) \\
&\leq |S| \cdot \frac{\gamma}{2k} + \epsilon \cdot \frac{1+\epsilon}{1-\epsilon} \cdot \gamma
\,,
\end{align*}
where the inequality $f(X) \leq \frac{1+\epsilon}{1-\epsilon}\gamma$ follows because  $f(X) \leq  f(\sol_t)/(1-\epsilon) \leq \frac{1+\epsilon}{1-\epsilon} \cdot \gamma$.
\end{proof}

%% file: appendix_precomp.tex
\section{Missing Proof from Section~\ref{sec:precomp}}
\label{app:precomputation}

\robustfromdynamic*
\begin{proof}
We will first set up some notation. 
The algorithm of \cite{lattanzi2020fully} creates 
several buckets $\{A_{i,\ell}\}_{ i\in [R], \ell \in [T]}$ where $R = O(\log k/\epsilon)$ is the number of thresholds  
and $T = O(\log n)$ is the number of levels. 
A solution is constructed using these buckets by iteratively selecting $S_{i,\ell}$ from $A_{i,\ell}$ starting from the smallest $(i,\ell) = (0,0)$ to the largest $(i,\ell) = (R,T)$.
The buffer at level $\ell$ is denoted by $B_\ell$.
Each set $S_{i,\ell}$ is constructed by iteratively selecting elements uniformly at random (without replacement) from the corresponding bucket.
More precisely, given indices $i,\ell$, the algorithm adds elements $e$ to the solution $S_{i,\ell}$ one by one in line $7$ of Algorithm $6$ in \cite{lattanzi2020fully} only if $f(e | S)\geq \tau_i \geq \gamma/2k$.

We will now show how to construct $(Q,R)$
from the $\dynamic$ data structure $A$.
Note that we will extract $(Q,R)$ by observing the evolution of the data-structure $A$.
Hence, we do not need to make any additional oracle queries in order to extract $(Q,R)$.


The following algorithm gives a procedure for extracting $(Q,R)$ by observing the actions of the peeling algorithm (Algorithm 6) in \cite{lattanzi2020fully}.

\begin{algorithm}[H]
\caption{\textsc{Robust1FromDynamic}}
\setstretch{1.08}
\hspace*{\algorithmicindent} \textbf{Input:} function $f : 2^V \> \R$,  dynamic data-structure $A$, constraint $k$, deletion parameter $d$, parameter $\epsilon$
\begin{algorithmic}[1]
\State $Q \leftarrow \emptyset$
\State $R \leftarrow \emptyset$
\For{$\ell \in [T]$}
\For{$i \in [R]$}
\If{$n/2^\ell > \frac{d}{\epsilon} + k$}
\State $Q \leftarrow  Q \cup S_{i,\ell}$
\Else
\State $R \leftarrow  R \cup A_{i,\ell} \cup B_\ell$
\EndIf

\EndFor
\EndFor
\State{\textbf{return} $Q, R$}
\end{algorithmic}
\label{alg:robustone}
\end{algorithm}

We will now show that the conditions required in 
Definition~\ref{def:robust_properties}
by the above algorithm.

\begin{itemize}
    \item \textbf{Size.} The fact that $|Q| \leq k$ follows using the $Q \subseteq \cup_{i, \ell} S_{i,\ell} = S$ where $S$ is the solution output by $\dynamic$ with $|S| \leq k$.

    We will now show that the size of $R = O(\frac{\log k}{\epsilon} \cdot (\frac{d}{\epsilon} + k))$.
    Let $\bar{\ell} \in [T]$ be the largest value such that the corresponding bucket size $n/2^{\bar{\ell}}$ is at least $ d/\epsilon + k$.
    Recall that
    \[
    R = \cup_{\ell < \bar{\ell}} B_{\ell} \cup \Big(\cup_{i \in [R]} A_{i,\ell} \Big)
        \,.
    \]
    We first have that
    \[
    |\cup_{\ell < \bar{\ell}} B_\ell| \leq \sum_{\ell < \bar{\ell}} n/2^{\ell} \leq 2 n/2^{\bar{\ell}-1} = O(\frac{d}{\epsilon} + k) 
        \,.
    \] 
    Secondly, 
    \[
    |\cup_{\ell < \bar{\ell}} \cup_{i \in [R]} A_{i,\ell}| \leq \sum_{\ell < \bar{\ell}} \sum_{ i\in [R]} n/2^{\ell} \leq \frac{\log k}{ \epsilon} \cdot 2 n/2^{\bar{\ell}-1} = O\Big(\frac{\log k}{\epsilon}\cdot \Big(\frac{d}{\epsilon} + k \Big) \Big) 
        \,.
    \] 
    Combining the above inequalities gives us the desired bound on the size of $R$.


    \item \textbf{Value.} The fact that $f(Q) \geq |Q| \gamma/2k$ follows directly using the property of the dynamic algorithm of \cite{lattanzi2020fully}
    an element $e$ is added to $S_{i, \ell}$ during a call to Bucket-Construct$(i, \ell)$ only if  $f_{S'}(e) \geq \tau_i$
    where $\tau_i \geq \gamma/2k$ and $S'$ is the partial solution. It is possible that some of the elements of the partial solution have been deleted 
    since the last time Bucket-Construct was 
    called on index $(i,\ell)$, but the marginal 
    contribution of $e$ can only increase with deletions due to submodularity.


We now show that,  if $|Q| < k$, then for any 
for any set $S \subseteq V \setminus R$ we have $f_Q(S ) < |S| \gamma / (2k) + \epsilon \gamma$.
Let $Q'$ denote the set of elements in $Q$ at the time when Level-Construct$(\bar{\ell})$ was called last. We have
    
    
    \begin{align*}
    f_Q(S) &= f(Q \cup S) - f(Q) \leq f(Q'\cup S) - f(Q) \\
    &\leq f(Q'\cup S) - f(Q') + f(Q') - f(Q) \\ 
    &\leq f(Q'\cup S) - f(Q') + f(Q') - (1-\epsilon) f(Q') \\ 
    &\leq \sum_{e \in S} f_{Q'}(e)  + \epsilon f(Q') \\
    &\leq |S| \cdot \frac{\gamma}{2k} + \epsilon \cdot \frac{1+\epsilon}{1-\epsilon} \cdot \gamma \,,
    \end{align*}
where the inequality $f(Q') \leq \frac{1+\epsilon}{1-\epsilon}\gamma$ follows because  $f(Q') \leq  f(Q)/(1-\epsilon) \leq \frac{1+\epsilon}{1-\epsilon} \cdot \gamma$,
and the inequality $f_{Q'}(e) \leq \gamma/2k$ follows by considering two cases: (1) $e$ was inserted before the latest time Level-Construct$(\bar{\ell})$ was called in which case $e$ cannot have high marginal contribution to $Q'$ as otherwise it would have been inserted in $Q'$, (2) $e$ was inserted after the latest time Level-Construct$(\bar{\ell})$ was called in which case either $e$ has low marginal contribution to $Q'$ or it would be included in the set $R$.





\item \textbf{Robustness.}  For any set $D \subseteq V$ such that $|D| \leq d$, $\E_Q[f(Q \setminus D)] \geq (1-\epsilon)  f(Q).$ 
We have that 
\[
f(Q) = \sum_{\ell \geq \bar{\ell} } \sum_{i \in [R]} \sum_{j \in [|S_{i,\ell}|]} f_{S_{ \ell, i,j}^{partial}}(e_{\ell, i,j})
    \,.
\]
where $e_{\ell, i,j}$ is the $j$-th element that 
was added to $S_{i,\ell}$ and $S_{ \ell, i,j}^{partial}$ is the partial solution before the $j$-th element was added.
We have that 
\[
\sum_{\ell \geq \bar{\ell} } \sum_{i \in [R]} |S_{i,\ell}|\tau_i \geq f(Q) \leq (1+\epsilon) \sum_{\ell \geq \bar{\ell} } \sum_{i \in [R]}|S_{i,\ell}|\tau_i
\]
Since, each element in $S_{i,\ell}$ is selected
from a set of size at least $d/\epsilon$,
it can only be deleted with probability at most $\epsilon$.
We have that
\begin{align*}
\E[f(Q\setminus D)] &= \sum_{\ell \geq \bar{\ell} } \sum_{i \in [R]} \sum_{j \in [|S_{i,\ell}|]} \E[\1[e_{\ell, i,j} \not\in D] \cdot f_{A_{ \ell, i,j}^{partial}}(e_{\ell, i,j}) ]\\
    &\geq \sum_{\ell \geq \bar{\ell} } \sum_{i \in [R]} \sum_{j \in [|S_{i,\ell}|]} \E[\1[e_{\ell, i,j} \not\in D] \cdot f_{S_{ \ell, i,j}^{partial}}(e_{\ell, i,j} )] \\
    & \geq \sum_{\ell \geq \bar{\ell} } \sum_{i \in [R]}\sum_{j \in [|S_{i,\ell}|]} \Pr(e_{\ell, i,j} \not\in D) \cdot \tau_i \\
    & \geq  (1+\epsilon) \sum_{\ell \geq \bar{\ell} } \sum_{i \in [R]}|S_{i,\ell}|\tau_i \\
    &\geq \frac{1+\epsilon}{1-\epsilon} f(Q). \qedhere
\end{align*}
\end{itemize}
\end{proof}